\documentclass{article}

\usepackage[a4paper,margin=3cm]{geometry}

\usepackage{hhw-arxiv}

\title{Reasoning About Higher-Order Relational Specifications}
\author{
  Yuting Wang\\
  University of Minnesota, USA \\
  \texttt{yuting@cs.umn.edu}
  \and
  Kaustuv Chaudhuri \\
  INRIA, France \\
  \texttt{kaustuv.chaudhuri@inria.fr}
  \and
  Andrew Gacek \\
  Rockwell Collins, USA \\
  \texttt{andrew.gacek@gmail.com}
  \and
  Gopalan Nadathur \\
  University of Minnesota, USA \\
  \texttt{gopalan@cs.umn.edu}
}

\begin{document}

\maketitle

\begin{abstract}
  The logic of hereditary Harrop formulas (\HHw) has proven useful for
  specifying a wide range of formal systems that are commonly presented via
  syntax-directed rules that make use of contexts and side-conditions.
  The two-level logic approach, as implemented in the Abella theorem prover,
  embeds the \HHw specification logic within a rich reasoning logic that
  supports inductive and co-inductive definitions, an equality predicate, and
  generic quantification.
  Properties of the encoded systems can then be proved through the embedding,
  with special benefit being extracted from the transparent correspondence
  between \HHw derivations and those in the encoded formal systems.
  The versatility of \HHw relies on the free use of nested implications, leading
  to dynamically changing assumption sets in derivations.
  Realizing an induction principle in this situation is nontrivial and the
  original Abella system uses only a subset of \HHw for this reason.
  We develop a method here for supporting inductive reasoning over all of \HHw.
  Our approach relies on the ability to characterize dynamically changing
  contexts through finite inductive definitions, and on a modified encoding of
  backchaining for \HHw that allows these finite characterizations to be used in
  inductive arguments.
  We demonstrate the effectiveness of our approach through examples of
  formal reasoning on specifications with nested implications in an extended
  version of Abella.
\end{abstract}

\section{Introduction}
\label{sec:introduction}

\noindent%
We are concerned in this paper with the task of reasoning about formal
systems such as programming languages, proof systems and process calculi.
The data objects that are of interest within such systems
often embody binding constructs.
Higher-order abstract syntax (\HOAS) provides an effective means
for representing such structure.
In an \HOAS representation, which is based on using a well-calibrated
\lcalculus as a metalanguage, the binding structure of object language
expressions is encoded using abstractions in \lterms.
For example, consider an object language that is itself a \lcalculus.
Letting \ktm be a type for the representation of these terms, their \HOAS
encoding can be built around two constructors,
$\kapp : \ktm \to \ktm \to \ktm$ and $\kabs : (\ktm \to \ktm) \to \ktm$:
the object term $\LAM x. \LAM y. y~x$ would, for instance, be
represented as $\kabs~(\LAM x. \kabs~(\LAM y. \kapp~y~x))$.
Observe that there is no constructor for variables in this encoding; object-level
variables are directly represented by the variables of the meta-language, bound
by an appropriate abstraction.
The virtue of \HOAS is that if the metalanguage is properly chosen, \ie, if it
incorporates $\lambda$-conversion but is otherwise weak in a computational
sense, then it provides a succinct and logically precise treatment of
object-level operations such as substitution and analysis of binding structure.

Formal systems are usually defined by the relations that hold between the
data objects that constitute them.
Such relations are conveniently presented through syntax-directed rules.
When they pertain to data embodying binding structure, these
specifications naturally tend to be higher-order, \ie, their rule-based
presentation involves the use of contexts.
Moreover, these contexts can contain conditional assertions whose use may
require the construction of sub-derivations.
Towards understanding this issue, consider the alternative notation for \lterms
due to De Bruijn in which bound variables are not named and their occurrences
are represented instead by indexes that count the abstractions up to the one
binding them \cite{debruijn72}.
Using the type \kdtm for the representation of \lterms in this form, we can
encode them via  the constructors $\kdvar : \knat \to \kdtm$ (for variables),
$\kdapp : \kdtm \to \kdtm \to \kdtm$ and $\kdabs : \kdtm \to \kdtm$.
Now, there is a natural bijection between the named and nameless representation
of \lterms.
Writing $\G |- \hodbrel{m,h,d}$ to denote the correspondence between the
\HOAS-encoded term $m$ that occurs at \emph{depth} $h$ (\ie, under $h$
$\lambda$-abstractions) and the De Bruijn term $d$ where $\G$ determines the
mapping between free variables in the two representations, we can define this
relation via these rules:
\begin{gather}
  \linfer{
    \G |- \hodbrel{\kapp~m~n,h,\kdapp~d~e}
  }{
    \G |- \hodbrel{m,h,d} & \G |- \hodbrel{n,h,e}
  }
  \label{rule:hodbrel-app}
  \\[1ex]
  \linfer{
    \G |- \hodbrel{\kabs~(\LAM x. m), h, \kdabs~d}
  }{
    \G, \ALL i, k. ((h + k = i) \supset \hodbrel{x, i, \kdvar~k})
    |- \hodbrel{m,h + 1,d}
  }
  \label{rule:hodbrel-abs}
  \\[1ex]
  \linfer{
    \G |- \hodbrel{x, i, \kdvar~k}
  }{
    \ALL i, k. ((h + k = i) \supset \hodbrel{x, i, \kdvar~k}) \in \G
    &
     |- h + k = i
  }
  \label{rule:hodbrel-var}
\end{gather}
The rule for relating applications is straightforward.
To relate $\kabs~(\LAM x. m)$ to a De Bruijn term at depth $h$,
we must relate each occurrence of $x$ in $m$, which must be at a depth $h +
k$ for some $k > 0$, to the De Bruijn term $\kdvar~k$.
To encode this correspondence, the context is extended in the premise of rule
(\ref{rule:hodbrel-abs}) with a (universally quantified) implicational formula.
Note also that this rule carries with it the implicit assumption that the name
$x$ used for the bound variable is fresh to $\G$, the context for the concluding
judgment.
Eventually, when the \HOAS term on the right of $|-$ is a variable, rule
(\ref{rule:hodbrel-var}) provides the means to complete the derivation by using
the relevant assumption from $\G$.
Observe that the use of this rule entails a construction of an auxiliary
derivation for $|-\, h + k = i$.

Our ultimate interest is in reasoning about such higher-order relational
specifications.
For example, we might be interested in showing that the relation that we have
defined above identifies a bijective mapping between the two
representations of \lterms.
One part of establishing this fact is proving that the relation is deterministic
from left to right, \ie, that every term in the named notation is related to at
most one term in the nameless notation.
Writing $\oseq{\G\, |-\, \hodbrel{m,h,d}}$ to denote derivability of
the judgment $\G\, |-\, \hodbrel{m,h,d}$ by virtue of the rules
(\ref{rule:hodbrel-app}), (\ref{rule:hodbrel-abs}) and (\ref{rule:hodbrel-var}),
this involves providing a proof for the following assertion:
\begin{gather}
  \ALL \G, m, h, d, e.
  \oseq{\G |- \hodbrel{m,h,d}} \rimp
  \oseq{\G |- \hodbrel{m,h,e}} \rimp
  d = e.
  \label{stmt:hodb-det-ltr-bad}
\end{gather}
Note that $\ALL$ and $\rimp$ in (\ref{stmt:hodb-det-ltr-bad}) are logical
constants at the reasoning level in contrast to the ones in
(\ref{rule:hodbrel-app}) -- (\ref{rule:hodbrel-var}) that are at the object
level.
Such a proof must obviously be based on an analysis of derivability using the
rules that define the relation.
To formalize such reasoning, we need a logic that can encode these rules
in a way that allows case analysis to be carried out over their structure.
Furthermore, the logic must embody an induction principle since proofs of general
theorems of the kind we are interested in must be inductive over the structure
of object-level derivations.
A particular difficulty in articulating such inductive arguments relative to
higher-order relational specifications is that they may need to take into
account derivations in the object system that rely on hypotheses in changing
contexts.
For example, a proof of (\ref{stmt:hodb-det-ltr-bad}) must accommodate the fact
that $\G$ can be dynamically extended in a derivation of $\G |- \hodbrel{m,h,d}$
and that the particular content of $\G$ influences the derivation in the
variable case via the rule (\ref{rule:hodbrel-var}).

In this paper we develop a framework that provides an elegant solution to this
reasoning problem.
Formally, our framework is a realization of the two-level logic
approach~\cite{gacek12jar,mcdowell02tocl}, which is based on embedding a
\emph{specification logic} inside a \emph{reasoning logic}.
Within this setup, we take our specification logic to be that of hereditary
Harrop formulas (\HHw).
This logic extends the well-known logic of Horn clauses essentially by employing
simply typed \lterms as a means for representing data objects and by permitting
universal quantification and implications in the bodies of clauses.
As such, it provides an excellent basis for encoding rule-based higher-order
specifications over \HOAS representations~\cite{miller12proghol}.
Moreover, these formulas can be given a proof-theoretic interpretation that
simultaneously is complete with respect to intuitionistic logic and reflects
the structure of derivations based on the object-level rules they encode.
For the reasoning logic we use the system \Gee from~\cite{gacek11ic}.
This logic permits atomic predicates to be defined through clauses in a way
that allows case analysis based reasoning to be carried out over them.
The treatment of definitions in \Gee can also be specialized to interpret them
inductively.
The capability for formally proving properties about relational specifications
is realized in this setting by first encoding \HHw derivability in \Gee via an
inductive definition and then using this encoding to reflect reasoning based on
object-level rules into reasoning over \HHw derivations that formalize these
rules.

The two-level logic approach has previously been implemented in the Abella
system and has been used successfully in several reasoning tasks~\cite{gacek09phd}.
However, the original version of Abella uses a fragment of \HHw that is capable
of treating syntax-directed rules only when the dynamic additions to their
contexts is restricted to atomic formulas.
There is an inherent difficulty in structuring the reasoning
when contexts can be extended with formulas having an implicational structure.
For example, as already noted, case analysis over the derivation of $\G |-
\hodbrel{m,h,d}$ in a proof of (\ref{stmt:hodb-det-ltr-bad}) must take into
account the fact that the derivation can proceed by using a hypothesis that was
dynamically added to $\G$.
Without well-defined constraints on $\G$, it is difficult to predict how such
hypotheses might be used and indeed the assertion may not even be
true.

In the example under consideration, there is an easy resolution to the dilemma
described above.
We are not interested in proving assertion (\ref{stmt:hodb-det-ltr-bad}) for
arbitrary $\G$ but only for those $\G$s that result from additions made through
the rule (\ref{rule:hodbrel-abs}).
The elements of $\G$ must therefore all be of the form $\ALL i, k. ((h + k = i)
\supset \hodbrel{x, i, \kdvar~k})$ where $h$ is some depth and $x$ is some
variable not otherwise present in $\G$.
Moreover, the use of such assumptions in derivations can occur only through rule
(\ref{rule:hodbrel-var}) that is in fact another instance of a backchaining step
that is manifest explicitly in the rules (\ref{rule:hodbrel-app}) and
(\ref{rule:hodbrel-abs}).
Thus, the structure of $\G$ can be encoded into an inductive definition in
$\Gee$ and treated in a finitary fashion by the machinery that $\Gee$ already
provides for reasoning about backchaining steps.

The key insight underlying this paper is that the above observation generalizes
cleanly to other reasoning situations that involve contexts with
higher-order hypotheses.
Concretely, the contexts that need to be considered in these situations are
completely determined by the additions that can be made to them.
Further, the structure of such additions must already be manifest in the
original specifications and can therefore always be encapsulated in an inductive
definition.
To take advantage of this observation we modify the encoding of \HHw derivations
in Abella to support reasoning also over the backchaining steps that result from
using dynamically added assumptions.
We then demonstrate the power of this extension through its use in explicitly
proving the bijectivity property discussed above as well as another non-trivial
property about paths in \lterms and their relation to reduction.
These exercises also show the benefits of using a logic for specifications:
the meta-theoretic properties of this logic
greatly simplify the reasoning process.

In summary, we make three contributions through this work: we propose a
methodology for reasoning about higher-order relational specifications, we
present an implemented system for supporting this methodology and we show
its effectiveness through actual reasoning tasks.
The framework we describe exploits the \HOAS representation style to structure
and simplify the reasoning process.
To the best of our knowledge, the only other systems that use such an approach
to similar effect are Twelf~\cite{pfenning99cade} and
Beluga~\cite{pientka10ijcar}.
In contrast to these systems, the one we develop here provides a rich language
for stating meta-theoretic properties of specifications and an explicit logic
for articulating their proofs.
We elaborate on these comparisons in a later section.

The rest of the paper is structured as follows.
In the next two sections, we present the specification logic \HHw, the reasoning
logic \Gee, and the two-level logic approach that is built out of their
combination.
\secref*{paths} illustrates the use of the resulting framework and the
associated methodology for a novel and non-trivial example.
The focus in this example is on specifications that have a rich higher-order
character and on showing how context definitions and context relations can be
used to structure and realize the reasoning process.
The last two sections discuss related work and conclude the paper by providing a
perspective on its technical contributions.

The extended Abella system that is the outcome of this work is available
at~\cite{abella-hhw}.
Besides the examples described in this paper, this version of Abella also
contains a number of other examples of reasoning about higher-order relational
specifications that illustrate our approach.

\section{The Specification Logic}
\label{sec:hhw}

\noindent%
In this section, we present the specification logic \HHw, show how it can be
used to encode rule-based descriptions, and discuss some of its meta-theoretic
properties that turn out to be useful in reasoning about specifications
developed in it.

\long\def\hhwrules{
\begin{figure}[!t]
  Goal reduction rules
  \begin{smallgather}
    \linfer[{\simp}_R]{
      \Si ; \Th ; \G  |- F \simp G
    }{
      \Si ; \Th ; \G, F |- G
    }
    \qquad
    \linfer[\sandR]{
      \Si ; \Th ; \G |- G_1 \sand G_2
    }{
      \Si ; \Th ; \G |- G_1 &
      \Si ; \Th ; \G |- G_2
    }
    \\[1ex]
    \linfer[{\sfall}_R]{
      \Si ; \Th ; \G |- \sfall_\tau G
    }{
      (c \notin \Si)
      &
      \Si, c{:}\tau ; \Th ; \G |- (G~c)
    }
  \end{smallgather}

  \vspace{-1ex}
  Backchaining rules
  \begin{smallgather}
    \linfer[{\simp}_L]{
      \Si ; \Th ; \G, \foc{G \simp\mkern -2mu F} |- A
    }{
      \Si ; \Th ; \G |- G
      &
      \Si ; \Th ; \G, \foc{F} |- A
    }
    \quad
    \linfer[\sandL]{
      \Si ; \Th ; \G, \foc{F_1 \sand F_2} |- A
    }{
      \Si ; \Th ; \G, \foc{F_i} |- A
    }
    \\[1ex]
    \linfer[\sfall_L]{
      \Si ; \Th ; \G, \foc{\sfall_\tau F} |- A
    }{
      \Si |- t : \tau &
      \Si ; \Th ; \G, \foc{(F~t)} |- A
    }
  \end{smallgather}

  \vspace{-1ex}
  Structural rules
  \vspace{-1ex}
  \begin{smallgather}
    \linfer[\rn{match}]{
      \Si ; \Th ; \G, \foc{A} |- A
    }{}
    \\[1ex]
    \linfer[\rn{prog}]{
      \Si ; \Th ; \G |- A
    }{
      (F \in \Th) &
      \Si ; \Th ; \G, \foc{F} |- A
    }
    \quad
    \linfer[\rn{dyn}]{
      \Si ; \Th ; \G |- A
    }{
      (F \in \G) &
      \Si ; \Th ; \G, \foc{F} |- A
    }
  \end{smallgather}
  \vspace{-0.5cm}
  \caption{Rules for \HHw. In ${\sand}_L$, $i \in \{1, 2\}$.}
  \label{fig:hhw-rules}
  \vspace{-0.35cm}
\end{figure}}

\subsection{The {\large\pmb\HHw} Proof System}
\label{sec:hhw-proof-system}

\noindent%
The logic \HHw of hereditary Harrop formulas is a predicative fragment of
Church's Simple Theory of Types~\cite{church40} whose expressions are simply
typed \lterms.
Types are built freely from primitive types, which must include the type \omic
of formulas, and the function type constructor $\to$.
Terms are built from a user-provided signature of typed constants, and are
considered identical up to $\alpha\beta\eta$-conversion.
We write $\Si |- t : \tau$ to denote that $t$ is a well-formed term of type
$\tau$ relative to $\Si$.
Well-formed terms of type $\omic$ relative to $\Si$ are called
\emph{$\Si$-formulas} or just \emph{formulas} when $\Si$ is implicit.

Logic is introduced into this background via a countable family of constants
containing: ${\simp},{\sand} : \omic \to \omic \to \omic$ (written infix, and
associating to the right and left, respectively), and for every type $\tau$ not
containing \omic, the (generalized) universal quantifier ${\sfall_\tau} : (\tau
\to \omic) \to \omic$.
An atomic formula, denoted by $A$ possibly with a subscript, is one that does
not have a logical constant as its head symbol.
We use the abbreviations $\SFALL x{:}\tau. F$ for $\sfall~(\LAM x{:}\tau. F)$,
$\SFALL x_1{:}\tau_1, \dotsc, x_n{:}\tau_n. F$ for $\SFALL x_1{:}\tau_1. \dotsc
\sfall x_n{:}\tau_n. F$, and $\SFALL \bar x{:}\bar \tau. F$ where $\bar x = x_1,
\dotsc, x_n$ and $\bar\tau = \tau_1, \dotsc, \tau_n$ for $\SFALL x_1{:}\tau_1.
\dotsc \SFALL x_n{:}\tau_n. F$.
We will omit the types when they are irrelevant or can be inferred from context.
Finally, we will often write $G \sif F$ (with ``$\sif$'' associating to the left
and pronounced ``if'') to mean $F \simp G$.

\hhwrules

The \HHw proof system has two kinds of sequents:
\begin{quote}
  \begin{tabular}{l@{\qquad}l}
    $\Si ; \Th ; \G |- G$ & \emph{goal-reduction sequent} \\[1ex]
    $\Si ; \Th ; \G, \foc{F} |- A$ & \emph{backchaining sequent}
  \end{tabular}
\end{quote}
In these sequent forms, $\Si$ is a signature; $\G$ and $\Th$ are multisets of
$\Si$-formulas; $G$ is a $\Si$-formula and $A$ is an atomic $\Si$-formula.
The context $\Th$ is called the \emph{static context} because it contains a
finite and unchanging \HHw \emph{program}.
The context $\G$, called the \emph{dynamic context}, contains the assumptions
introduced during the goal reduction procedure, and can therefore grow.
The members of $\Th$ and $\G$ are called the \emph{static clauses} and the
\emph{dynamic clauses} respectively.

\figref*{hhw-rules} contains the inference rules of \HHw.
Reading the rules as a computation of premise sequents from goal sequents, the
\emph{goal reduction rules} decompose the goal on the right of $|-$ until it
becomes atomic.
The ${\simp}_R$ rule extends the dynamic context with the antecedent of the
implication, while the ${\sfall}_R$ rule extends the signature with a fresh
constant for the universally quantified variable.

Once the goal becomes atomic, the only rules that apply are the final two
structural rules that \emph{select} a backchaining clause.
The \rn{prog} rule selects a static clause, while the \rn{dyn} rule selects a
dynamic clause.
In either case, the premise is a backchaining sequent with the selected clause
indicated by $\foc{-}$.
The \HHw proof system does not prescribe a strategy for selecting clauses, so to
reason about \HHw derivations we will have to consider every possibility.

While the selected clause is non-atomic, the \emph{backchaining rules} are
used to reduce it.
The ${\simp}_L$ rule changes the selection to the succedent of the implication,
moving in the direction of the \emph{head} of the clause, and additionally
checks that the antecedent is derivable.
The ${\sand}_L$ rules change the selection to one of the operands of a $\sand$.
The ${\sfall}_L$ rule changes the selection to some instance of a universally
quantified clause.
When the selected clause has been reduced to atomic form, the corresponding
branch of the proof finishes by the rule \rn{match} which requires that the
atomic clause match the atomic goal.
Therefore, if the right hand side does not match, then this branch of the proof
is invalid and some choice made earlier in the proof needs to be revisited.

In the common case of a clause with the form $\SALL \bar x{:}\bar \tau. G_1
\simp \dotsm \simp G_n \simp A$, the \rn{match} rules and the backchaining rules
compose to give this derived rule:
\begin{smallgather}
  \infer{
    \Si ; \Th ; \G, \foc{\SALL \bar x{:}\bar \tau. G_1 \simp \dotsm \simp G_n \simp A}
    |- [\bar t/\bar x] A
  }{
    \Si |- \bar t : \bar \tau
    &
    \Si ; \Th ; \G |- [\bar t/\bar x] G_1
    & \dotsm &
    \Si ; \Th ; \G |- [\bar t/\bar x] G_n
  }
\end{smallgather}
This derived form can readily be seen as implementing the backchaining
procedure: the goal on the right of $|-$ is matched against the head of a
selected clause, and then new goals are generated corresponding to the body of
the clause.

\subsection{Example: HOAS vs. De Bruijn {\large\pmb\lterms}}
\label{sec:hhw-example}

\noindent%
As a concrete example of a higher-order relational specification in \HHw, let us
consider the example in the introduction of \lterms represented in two different
ways, one with higher-order abstract syntax (\HOAS) and the other using De
Bruijn indexes.
The signature of this specification consists of the following basic types:
$\knat$ (for natural numbers), $\ktm$ (for \HOAS terms) and $\kdtm$ (for De
Bruijn terms), together with the following constants.

\smallskip \noindent \begingroup \small%
\begin{tabular}{l@{\ }|@{\ }l@{\ }|@{\ }l}
  \hfil \knat & \hfil \HOAS (\ktm) & \hfil De Bruijn (\kdtm) \\
  \hline \vrule width 0pt height 1em
  $\kz : \knat$
  &
  $\kapp : \ktm \to \ktm \to \ktm$
  &
  $\kdapp : \kdtm \to \kdtm \to \kdtm$
  \\
  $\ks : \knat \to \knat$
  &
  $\kabs : (\ktm \to \ktm) \to \ktm$
  &
  $\kdabs : \kdtm \to \kdtm$
  \\
  & &
  $\kdvar : \knat \to \kdtm$
\end{tabular}
\endgroup \smallskip

The static context specifies two relations, $\kadd : \knat \to \knat \to \knat
\to \omic$ and $\hodb : \ktm \to \knat \to \kdtm \to \omic$, that define
addition relationally and relate the two encodings of terms at a given depth.
These relations are given by the following static clauses.
\def\Raddz{\ensuremath{R_{\text{addz}}}\xspace}
\def\Radds{\ensuremath{R_{\text{adds}}}\xspace}
\def\Rapp{\ensuremath{R_{\text{app}}}\xspace}
\def\Rabs{\ensuremath{R_{\text{abs}}}\xspace}
\begin{smallalign}
  & \kadd~\kz~X~X. \tag{\Raddz} \\[1ex]
  & \kadd~(\ks~X)~Y~(\ks~Z) \sif \kadd~X~Y~Z. \tag{\Radds} \\[1ex]
  & \hodb~(\kapp~M~N)~H~(\kdapp~D~E) \sif {} \\
  &\TAB \hodb~M~H~D \sand \hodb~N~H~E.
  \tag{\Rapp} \\[1ex]
  & \hodb~(\kabs~M)~H~(\kdabs~D) \sif {} \\
  &\TAB \SALL x. \hodb~(M~x)~(\ks~H)~D \sif {} \\
  &\TAB[2] \SALL i, k. \hodb~x~i~(\kdvar~k) \sif \kadd~H~k~i.
  \tag{\Rabs}
\end{smallalign}
The clauses are written using the standard convention of indicating variables
that are universally quantified using capital letters; that is, the clause
\Raddz stands for $\SALL X. \kadd~\kz~X~X$, \etc\@
The clauses \Rapp and \Rabs provide a transparent encoding of rules
(\ref{rule:hodbrel-app}) and (\ref{rule:hodbrel-abs}) relative to the \HHw proof
system.
Note especially the embedded implication in the body of \Rabs: as we see in more
detail in the example derivation below, when combined with the derived
backchaining and the goal reduction rules, this implication leads to proving a
sequent with an extended dynamic context that closely resembles the premise of
(\ref{rule:hodbrel-abs}).
There is no clause corresponding to rule (\ref{rule:hodbrel-var}); it will arise
from clauses in the dynamic context as part of the backchaining mechanism of
\HHw.

Let $\Si$ be the signature above and $\Th$ be $\Raddz$, $\Radds$, $\Rapp$,
$\Rabs$.
Let us try to show that the term $\LAM x. \LAM y. (y~x)$ corresponds to the De
Bruijn term $\LAM. \LAM. (1~2)$.
This amounts to proving the following \HHw sequent:
\begin{smallalign}
  \Si ; \Th ; {\emp} &|- {} \\
  \hodb~&(\kabs~(\LAM x. \kabs~(\LAM y. \kapp~y~x)))~\kz \\
  &(\kdabs~(\kdabs~(\kdapp~(\kdvar~(\ks~\kz))~(\kdvar~(\ks~(\ks~\kz))))))).
\end{smallalign}
The dynamic context is empty and the goal is atomic, so only the \rn{prog} rule
is applicable.
Selecting \Raddz or \Radds will fail because the heads are different predicates,
and selecting \Rapp will also fail because the first-argument of \hodb is \kabs,
which does not unify with \kapp.
Therefore, the only choice is backchaining \Rabs, which changes the proof
obligation to:
\begin{smallalign}
  \Si, x{:}\knat ; \Th ;
  &(\SALL i, k. \hodb~x~i~(\kdvar~k) \sif \kadd~\kz~i~k) |- {} \\
  \hodb~&(\kabs~(\LAM y. \kapp~y~x))~(\ks~\kz) \\
  &(\kdabs~(\kdapp~(\kdvar~(\ks~\kz))~(\kdvar~(\ks~(\ks~\kz))))).
\end{smallalign}
Attempting to backchain the new dynamic clause using \rn{dyn} will fail because
the new signature constant $x$ does not unify with \kabs.
Hence, the sole possibility that remains is back\-chaining \Rabs again, yielding:
\begin{smallalign}
  \Si, x,y{:}\knat ; \Th ;
  &(\SALL i, k. \hodb~x~i~(\kdvar~k) \sif \kadd~\kz~k~i), \\
  &(\SALL i, k. \hodb~y~i~(\kdvar~k) \sif \kadd~(\ks~\kz)~k~i) |- {} \\
  \hodb~&(\kapp~y~x)~(\ks~(\ks~\kz)) \\
  &(\kdapp~(\kdvar~(\ks~\kz))~(\kdvar~(\ks~(\ks~\kz)))).
\end{smallalign}
Now we can only backchain \Rapp to yield two new proof obligations, the first of
which is:
\begin{smallalign}
  \Si, x,y{:}\knat ; \Th ;
  &(\SALL i, k. \hodb~x~i~(\kdvar~k) \sif \kadd~\kz~k~i), \\
  &(\SALL i, k. \hodb~y~i~(\kdvar~k) \sif \kadd~(\ks~\kz)~k~i) |- {} \\
  \hodb~&y~(\ks~(\ks~\kz))~(\kdvar~(\ks~\kz)).
\end{smallalign}
The only clause that we can select for backchaining is the second dynamic clause
for $y$ using \rn{dyn}; none of the other clauses have a matching head.
This modifies the goal to:
\begin{smallalign}
  \Si, x,y{:}\knat ; \Th ;
  &(\SALL i, k. \hodb~x~i~(\kdvar~k) \sif \kadd~\kz~k~i), \\
  &(\SALL i, k. \hodb~y~i~(\kdvar~k) \sif \kadd~(\ks~\kz)~k~i) |- {} \\
  \kadd~&(\ks~\kz)~(\ks~\kz)~(\ks~(\ks~\kz)).
\end{smallalign}
This sequent is then proved by backchaining \Radds and \Raddz.
The other proof obligation is handled similarly.

\subsection{Meta-theorems of {\large\pmb\HHw}}
\label{sec:hhw-meta}

\noindent%
As a logic, \HHw possesses several properties that can be useful in analyzing
derivability and therefore in reasoning about specifications written in
it. The following meta-theorems will be specifically useful in the examples we
consider.

\begin{theorem}[meta theorems of \HHw] \label{thm:hhw-meta} \mbox{} %
  \begin{ecom}
  \item If $\Si ; \Th ; \G |- F$ and $\Si ; \Th ; \G, F |- G$ are derivable,
    then so is $\Si ; \Th ; \G |- G$ (cut).
  \item If $\Si |- t : \tau$ and $\Si, c{:}\tau ; \Th ; \G |- G$ (where $c$ is
    not free in $\Th$) is derivable, then so is $\Si ; \Th ; [t/c]\G |- [t/c]
    G$, where $[t/c]$ stands for the capture-avoiding substitution of $t$ for
    $c$ (instantiation).
  \item If $\Si ; \Th ; \G |- G$ is derivable, and $F \in \G$ implies $F \in
    \D$, then $\Si ; \Th ; \D |- G$ is also derivable (monotonicity).
  \end{ecom}
\end{theorem}

\begin{proof}
  Each theorem follows by a straightforward inductive argument.
  See also \thmref{hhw-meta-formal}. \qed
\end{proof}

\noindent%
A direct corollary of the monotonicity theorem is that weakening and contraction
are admissible for the dynamic context.
Observe that the static context $\Th$ never changes, even in the case of cut and
instantiation.
Obviously this theorem holds even if $\Th$ is empty, so a variant proof system
that combines the static and dynamic contexts into a single context will also
enjoy the same properties.
However, when reasoning about the specification of a computational system, we
are almost never interested in considering situations where the static rules of
the system change.

\section{The Two-Level Logic Approach}

\noindent%
We describe now the reasoning logic \Gee and outline the encoding
of \HHw in \Gee that underlies our particular use of the two-level logic
approach. We then illustrate the resulting framework by using it to
formalize and prove the bijectivity property of the relation between \HOAS and
De Bruijn representations of \lterms.

\label{sec:2ll}

\long\def\geerules{
\begin{figure}[!t]
  \small
  \begin{gather*}
    \linfer[\rn{id}]{
      \X ; \D, B ||- B'
    }{
      (B \eqvt B')
    }
    \qquad
    \linfer[\cut]{
      \X ; \D ||- C
    }{
      \X ; \D ||- B &
      \X ; \D, B'\mkern -4mu ||- C &
      (B \eqvt B')
    }
    \\[1ex]
    \linfer[{\rfall}\mkern -2mu_L]{
      \X ; \D, \rfall_\tau B ||- C
    }{
      \X, \Si, \NomC |- t : \tau &
      \X ; \D, B~t ||- C
    }
    \\[1ex]
    \linfer[{\rfall}\mkern -2mu_R]{
      \X ; \D ||- \rfall_\tau~B
    }{
      (h \notin \X) &
      (\bar c = \supp(B)) &
      \X, h ; \D ||- B~(h~\bar c)
    }
    \\[1ex]
    \linfer[{\rnbl}\mkern -3mu_L]{
      \X ; \D, \rnbl\mkern -3mu_\tau~B ||- C
    }{
      (a \in \NomC \setminus \supp(B)) &
      \X ; \D, (B~a) ||- C
    }
    \\[1ex]
    \linfer[{\rnbl}\mkern -3mu_R]{
      \X ; \D ||- \rnbl\mkern -3mu_\tau~B
    }{
      (a \in \NomC \setminus \supp(B)) &
      \X ; \D ||- (B~a)
    }
  \end{gather*}
  \vspace{-2em}
  \caption{Selected rules of \Gee.}
  \label{fig:gee-quant-rules}
  \vspace{-0.35cm}
\end{figure}
}

\subsection{The Reasoning Logic {\large\pmb\Gee}}
\label{sec:gee}

\noindent%
Specifications based on derivation rules are usually given a \emph{closed-world}
reading, where relations are considered to be characterized fully by the rules
that describe them.
For instance, the rules that assign simple types to \lterms can be used not only
to identify types with well-formed terms, but also to argue that a term such as
$\LAM x. x~x$ cannot be typed.
The \HHw logic can be used to realize only the positive part of such
specifications.
To completely formalize the intended meaning of rule-based specifications, we
use the logic \Gee~\cite{gacek11ic} that supports inductive fixed-point
definitions.

The basis for \Gee is also an intuitionistic and predicative version of Church's
Simple Theory of Types.
Types are determined in \Gee as in \HHw except that the type of formulas is
\prop rather than \omic.
We assume a fixed collection $\Si$ of logical and non-logical constants none of
whose members other than the ones mentioned below contains \prop in its argument
types.
The logical constants of \Gee consist initially of $\rtrue$ and $\rfalse$ of
type $\prop$; ${\rand}$, ${\ror}$ and ${\rimp}$ of type $\prop \to \prop \to
\prop$; for every type $\tau$ not containing $\prop$, the quantifiers
$\rfall_\tau$ and $\rexts_\tau$ of type $(\tau \to \prop) \to \prop$; and the
equality symbol $=_\tau$ of type $\tau \to \tau \to \prop$.
To provide the capability of reasoning about \emph{open} \lterms, which is
necessary in many arguments about \HOAS, \Gee also supports
\emph{generic} reasoning.
Specifically, for every type $\tau$ not containing \prop, \Gee includes an
infinite set of \emph{nominal constants} of type $\tau$, and a \emph{generic
  quantifier} ${\rnbl}\mkern -3mu_\tau$ of type $(\tau \to \prop) \to
\prop$~\cite{miller05tocl}.
Like with \HHw, we often omit types and adopt the usual syntactic conventions
for displaying the logical connectives.

The proof system for \Gee is presented as a sequent calculus with sequents of
the form $\X ; \D ||- C$ where $\D$ is a set of formulas (\ie, terms of type
$\prop$), $C$ is a formula, and $\X$ contains the free eigenvariables in
$\D$ and $C$.
The treatment of fixed-point definitions in \Gee results in the eigenvariables
being given an extensional interpretation; in other words, unfolding a
definition on the left may instantiate some of the eigenvariables and introduce
other eigenvariables.
We write $\X, \Si, \NomC |- t : \tau$ to mean that $t$ is a well-formed term of
type $\tau$ all of whose free variables, constants, and nominal constants are
drawn from the respective sets to the left of $|-$.
Here and elsewhere, we use $\NomC$ to denote the collection of all nominal
constants that we assume to be disjoint from the eigenvariables contained in
$\X$ and the (logical and non-logical) constants contained in the signature,
$\Si$.

\geerules

Nominal constants are used to simplify generic judgments in the course of proof
search.
A correct formalization of this idea needs two provisos: that quantifier scopes
be respected and that judgments that differ only in the names of nominal
constants be identified.
\figref*{gee-quant-rules} contains a few rules of \Gee that show how
these conditions are realized; the full system can be found
in~\cite{gacek11ic}.
The essential feature of nominal constants is \emph{equivariance}: two terms $B$
and $B'$ are considered to be equal, written $B \eqvt B'$, if they are
$\lambda$-convertible modulo a permutation of the nominal constants.
We write $\supp(B)$---called the \emph{support} of $B$---for the (finite)
collection of nominal constants occurring in $B$.
The rules for $\rnbl$ are the same on both sides of the sequent; in each case a
nominal constant that doesn't already exist in the support of the principal
formula is chosen to replace the $\rnbl$-quantified variable.
In the ${\rfall}_R$ rule of \figref{gee-quant-rules}, the eigenvariable is
\emph{raised} over the support of the principal formula; this is needed to
express permitted dependencies on these nominal constants in a situation where
later substitutions for eigenvariables will not be allowed to contain them.
Note, however, that nominal constants may be used in witnesses in the
${\rfall}\mkern -2mu_L$ rule.

To accommodate fixed-point definitions, \Gee is parameterized by sets of
\emph{definitional clauses}.
Each such clause has the form $\ALLx \bar x. (\NABx \bar z. A) \rdef B$ where
$A$ is an atomic formula (called the \emph{head}) whose free variables are drawn
from $\bar x$ and $\bar z$, and $B$ is an arbitrary formula (called the
\emph{body}) whose free variables are also free in $\NABx \bar z. A$.
Each clause partially defines a relation named by the predicate in the head.
In every definitional clause $\ALLx \bar x. (\NABx \bar z. A) \rdef B$, we
require that $\supp(\NABx \bar z. A)$ and $\supp(B)$ are both empty.
Consistency of \Gee also requires predicate occurrences in the body of a clause
to also satisfy certain \emph{stratification} conditions, explained
in~\cite{gacek11ic}.

\Gee also includes special rules for interpreting definitional clauses.
When an atom occurs on the right of a sequent, then any of the clauses with a
matching head may be used to replace the atom by the corresponding body of the
clause; in other words, clauses may be backchained.
Matching the head of a clause requires some care with regard to the
quantifiers.
To match the head of a clause $\ALLx \bar x. (\NABx \bar z. A) \rdef B$ against
the atom $A'$, we look for a collection of distinct nominal constants $\bar c$
and witness terms $\bar t$ that do not
contain any of the elements of $\bar c$ such that $[\bar t/\bar x, \bar c/\bar
z]A \eqvt A'$.
If these can be found, then $A'$ is replaced on the right by $[\bar
t/\bar x]B$.
When an atom $A$ occurs on the left in a sequent, for every clause and every
way of \emph{unifying} the head of the clause to that atom, a new premise is
created with the corresponding body added to the context.
This amounts to a \emph{case analysis} over the clauses in a definition.
Note that substitutions into the clause must respect the order of the $\ALLx$
and $\NABx$ quantifiers at its head and that different unifiers may result from
considering different distinct nominal constant instantiations for the $\NABx$
quantifiers.
Some of the eigenvariables may be instantiated in the premises thus created so
the eigenvariable context should be modified to reflect the resulting changes.

The final crucial component derived from \Gee that we use in this paper is the
ability to mark certain predicates as being \emph{inductive}, whereby the set of
clauses for that predicate is interpreted as a least fixed-point definition.
When deriving a sequent of the form
\begin{smallgather}
  \X ; \D ||- \ALLx \bar x.
  F_1 \rimp \dotsm \rimp A \rimp \dotsm \rimp F_n \rimp G
\end{smallgather}
by induction on the atom $A$, \Gee produces this premise:
\begin{smallgather}
  \X, \bar x ; \D,
  (\ALLx \bar x. F_1 \rimp \dotsm \rimp A^* \rimp
     \dotsm \rimp F_n \rimp G), \\
  F_1, \dotsc, A^@, \dotsc, F_n ||- G
\end{smallgather}
Here, $A^*$ and $A^@$ are simply annotated versions of $A$ standing for
\emph{strictly smaller} and \emph{equal sized} measures respectively.
If $A^@$ is unfolded using a definitional clause, the predicates in the body of
the corresponding clause are given the $^*$ annotation; thus, the inductive
hypothesis (containing $A^*$) only becomes usable after at least one unfolding
of $A^@$.
For each following use of induction, a new set of annotations is produced (\eg,
$^{**}$ and $^{@@}$).
This use of annotations is justified by using $\lambda \bar x. F_1 \rimp
\dotsm \rimp A \rimp \dotsm \rimp F_n \rimp G$ as inductive invariant in a more
general (and also more abstract) rule that codifies a least fixed-point
treatment of the definition of $A$.
A formal development of the connection and a correctness argument
can be found in~\cite{gacek09phd}.

\subsection{Encoding {\large\pmb\HHw} in {\large\pmb\Gee}}
\label{sec:encoding-hhw}

\begin{figure*}[!tb]
  \begin{smallgather}
    \begin{array}[t]{l@{\ }c@{\ }l}
      \kseq~L~(G_1 \sand G_2)
      &\rdef&
      \kseq~L~G_1 \rand \kseq~L~G_2
      \\
      \kseq~L~(F \simp G)
      &\rdef&
      \kseq~(F \kcons L)~G
      \\
      \kseq~L~(\sfall_\tau~G)
      &\rdef&
      \NABx x{:}\tau. \kseq~L~(G~x)
      \\[1ex]
      \kseq~L~A
      &\rdef&
      \katomic~A \rand
      \kmember{F}{L} \rand \kbc~L~F~A
      \\
      \kseq~L~A
      &\rdef&
      \katomic~A \rand
      \kwd{prog}~F \rand \kbc~L~F~A
      \\[2ex]
      \kbc~L~(F_1 \sand F_2)~A
      &\rdef&
      \kbc~L~F_1~A \ror
      \kbc~L~F_2~A
      \\
      \kbc~L~(G \simp F)~A
      &\rdef&
      \kseq~L~G \rand
      \kbc~L~F~A
      \\
      \kbc~L~(\sfall_\tau~F)~A
      &\rdef&
      \EX t{:}\tau. \kbc~L~(F~t)~A
      \\[1ex]
      \kbc~L~A~A
      &\rdef&
      \rtrue
    \end{array}
  \end{smallgather}
  \caption{Encoding of \HHw rules as inductive definitions in \Gee.}
  \label{fig:hhw-encoding}
\end{figure*}

\noindent%
The logic \Gee has the necessary ingredients to represent the \HHw proof system
as an inductive definition.
Formally, the type \omic of \HHw is imported as an \emph{uninterpreted} type in
\Gee.
Two \HHw formulas $H, G : \omic$ may be compared only for syntactic equality (or
unifiability) in \Gee; in other words, the \Gee formula $H = G$ does not check
for logical equivalence (in \HHw) of $H$ and $G$.
The connectives of \HHw are thus treated as \emph{constructors} of \omic, so $(F
\simp G) = (F' \simp G')$ in \Gee would entail that $F = F'$ and $G =G'$ because
of congruence (\ie, injectivity of constructors), and $(F \simp G) = G$ would
not hold, even if $F$ were known to be derivable in \HHw, because $F \simp G$
and $G$ are not unifiable.

To encode \HHw sequents in \Gee, we first note that \Gee and \HHw share the same
type system.
The \HHw signature can therefore be imported transparently into \Gee, so the
signatures of \HHw sequents will not be explicitly encoded.
The contexts of \HHw are represented in \Gee as lists of \HHw formulas (\ie,
lists of terms of type \omic).
The type \olist with constructors $\knil : \olist$ and $(\kcons) : \omic \to
\olist \to \olist$ is used for these lists, and, per tradition, the $\kcons$
constructor is written infix.
Membership in a context is defined inductively as a predicate $\kwd{member} :
\omic \to \olist \to \prop$ with these clauses:
\begin{smallalign}
  \kmember{E}{(E \kcons L)} &\rdef \rtrue \\
  \kmember{E}{(F \kcons L)} &\rdef \kmember{E}{L}.
\end{smallalign}
Observe that the two clauses have overlapping heads; there will be as many ways
to show $\kmember{E}{L}$ as there are occurrences of $E$ in $L$.
This validates the view of \HHw contexts as multisets.

The sequents of \HHw are then encoded in \Gee using the predicates \kseq and
\kbc.
\begin{quote}
  \begin{tabular}{l@{\qquad}l@{\qquad}l}
    \hfil \HHw & \hfil \Gee & \hfil notation \\
    \hline \vrule width 0pt height 1em
    $\Th ; \G |- G$ & $\kseq~L~G$ & $\oseq{L |- G}$ \\
    $\Th ; \G, \foc{F} |- A$ & $\kbc~L~F~A$ & $\oseq{L, \foc{F} |- A}$
  \end{tabular}
\end{quote}
Here, $L$ is an \olist representation of $\G$.
The third column contains a convenient and evocative notation for the equivalent
\Gee atom in the second column; we shall often use this notation in the rest of
this paper.
Note that while the \HHw contexts are unordered multisets, the \olist
representations are ordered.
This is not a limitation because we will always reason about the contexts using
$\kwd{member}$.

The static program clauses in $\Th$ are not part of the \Gee encoding of
sequents.
Rather, we use the inductively defined predicate $\kprog : \omic \to \prop$ that
has one clause of the form $\kprog~F \rdef \rtrue$ for each $F \in \Th$.

The rules of the \HHw proof system in \figref{hhw-rules} are used to build
mutually inductive definitions of the $\kseq$ and $\kbc$ predicates.
This definition is depicted in \figref{hhw-encoding}; each clause of the
definition corresponds to a single rule of \HHw.
The goal reduction rules are systematically translated into the clauses, the
only novelty being that universally quantified variables of the specification
logic are represented as nominal constants in \Gee using the $\nabla$
quantifier.
This use of $\nabla$ is necessary because the encoding must completely
characterize provability in \HHw.
In particular, in \HHw the sequent $\emp ; (\SALL x. \kwd{eq}~x~x) |- \SALL y,
z. \kwd{eq}~y~z$ is \emph{not} derivable, meaning that the \Gee formula
$\kseq~((\SALL x. \kwd{eq}~x~x) \kcons \knil)~((\SALL y, z. \kwd{eq}~y~z) \rimp
\rfalse)$ should be true.
This is achievable since it unfolds to $(\NAB y, z. \kseq~((\SALL x.
\kwd{eq}~x~x) \kcons \knil)~(\kwd{eq}~y~z)) \rimp \rfalse$.
As a point of comparison, if we were to use this clause instead:
\begin{smallgather}
\kseq~L~(\sfall_\tau\ G) \rdef \ALL x{:}\tau. \kseq~L~(G~x)
\end{smallgather}
then the non-derivability property of the \HHw sequent above, now encoded as
$(\ALL y, z. \kseq~((\SALL x. \kwd{eq}~x~x) \kcons \knil)~(\kwd{eq}~y~z)) \rimp
\rfalse$, would not be true. (In particular, the antecedent is satisfiable in
models with only a single inhabitant.)

The backchaining rules of \HHw are encoded as clauses of \kbc in a
straightforward manner.

For the structural rules of \HHw, we have to enforce the invariant that the
right hand side of the sequent is atomic.
This is achieved by means of a predicate $\katomic : \omic \to \prop$
defined by the following clause:
\begin{smallalign}
  \katomic~F
  &\rdef \bigl(\ALL G. (F = \sfall_\tau\ G) \rimp \rfalse\bigr) \\
  &\rand \bigl(\ALL G_1, G_2. (F = (G_1 \sand G_2)) \rimp \rfalse\bigr) \\
  &\rand \bigl(\ALL G_1, G_2. (F = (G_1 \simp G_2)) \rimp \rfalse\bigr).
\end{smallalign}
Effectively, \katomic characterizes atomic formulas negatively by saying
that an atomic formula cannot be constructed with a \HHw connective.
It is important to note that there is a small issue with all three of \kseq,
\kbc, and \katomic: they treat $\sfall_\tau\ G$ as if it were a single
object, but, since the reasoning and specification logics share the type system,
it actually stands for all instances for the type $\tau$.
To keep these definitions finite, we would require polymorphism, which \Gee
currently lacks.
In the Abella implementation, therefore, these definitions are treated
specially.
Note that the meta-theory of \Gee does not require that inductive definition
have finitely many clauses, so even an infinitary interpretation of the clauses
of \figref{hhw-encoding}, as was done in~\cite{gacek12jar}, is compatible
with our approach.

The faithfulness of our encoding allows us to state and prove known properties
of \HHw in \Gee.
For example, the meta-theoretic properties discussed in
\thmref{hhw-meta} have the following counterparts relative to the
encoding in \Gee.
Having proved them in \Gee, we can use the \cut rule to invoke them as lemmas in
arguments concerning particular specifications.

\begin{theorem}\label{thm:hhw-meta-formal}%
  The earlier discussed meta-theoretic properties of \HHw are validated by their
  encoding in \Gee.
  In other words, each of the following is provable in \Gee.
  \begin{ecom}
  \item $ \ALL L, F, G. \oseq{L |- F} \rimp \oseq{L, F |- G} \rimp \oseq{L |-
      G}$ (cut).
  \item $\ALL L, G. \NABx x. \oseq{L~x |- G~x} \rimp \ALL t. \oseq{L~t |- G~t}$
    (instntiation).
  \item $\ALL L, L', G. \oseq{L |- G} \rimp (\ALL F. \kmember{F}{L} \rimp
    \kmember{F}{L'}) \rimp \oseq{L' |- G}$ (monotonicity).
  \end{ecom}
\end{theorem}

\begin{proof}
  These are fairly straightforward inductive theorems of \Gee.
  We have proved them formally in the Abella~\cite{abella-hhw} implementation of
  \Gee; the proofs can be found in the file \url{hh_meta.thm}.
  %
\end{proof}

\subsection{Example: HOAS vs. De Bruijn Revisited}
\label{sec:hodb}

\noindent%
We are now in a position to formally verify that the relation presented in the
introduction between the encodings of the named and nameless representations of
\lterms actually specifies an isomorphism.
We do this by showing that its rendition in \HHw described in
\secref{hhw-example} is deterministic in both its first and
third arguments.
As expected, we work within \Gee with the encoding of \HHw described in the
previous section.
We also assume that the (static) clauses \Raddz, \Radds, \Rapp, and \Rabs have
been reflected into the definition of \kprog in this context.

As mentioned in the introduction, we will need to finitely characterize the
possible dynamic context extensions during the derivation of \hodb.
The inductive definition of these dynamic contexts of \hodb has the following
pair of clauses.
\begin{smallalign}
  &\kctx~\knil \rdef \rtrue \\
  \Bigl(\NABx x. &\kctx~((\SALL i, k.
  \hodb~x~i~(\kdvar~k) \sif \kadd~H~k~i
  ) \kcons L) \Bigr) \rdef \kctx~L.
\end{smallalign}
As usual, the capitalized variables $H$ and $L$ are universally quantified over
the entire clause.
Note the occurrence of $\nabla\mkern -2mu x$ at the head of the second clause of
the definition: it guarantees that $x$ does not occur in $H$ or $L$.
Therefore, in any $L$ for which $\kctx~L$ holds, it must be the case that
there is exactly one such dynamic clause for each such $x \in \supp(L)$.
It is easy to establish this fact in terms of a pair of lemmas.

The first of these lemmas characterizes the dynamic clauses.
\begin{smallgathertagged}
  \begin{split}
   \ALL L, E. \kctx~&L \rimp \kmember{E}{L} \rimp \\
  \TAB \EX x, H. E &= \big(\SALL i, k. \kldb~x~i~(\kdvar~k\bigr)
    \sif \kadd~H~k~i\bigr) \\
  &\rand \kname{x}.
  \end{split}
  \label{lab:hodb-ctx-inv}
\end{smallgathertagged}
Here, $\kname{x}$ is a predicate that asserts that $x$ is a nominal constant;
this predicate can be defined in \Gee with the clause $(\NABx x. \kname{x})
\rdef \rtrue$.
To prove (\ref{lab:hodb-ctx-inv}), we proceed by induction on the first
hypothesis, $\kctx~L$.
As mentioned in \secref{gee}, this is achieved by assuming a new
\emph{inductive hypothesis} $\IH$:
\begin{smallgather}
  \begin{split}
   \ALL L, E. (\kctx~&L)^* \rimp \kmember{E}{L} \rimp \\
  \TAB \EX x, H. E &= \big(\SALL i, k. \kldb~x~i~(\kdvar~k\bigr)
    \sif \kadd~H~k~i) \\
  &\rand \kname{x}.
  \end{split}
  \tag{\IH}
\end{smallgather}
Moreover, the proof obligation is modified to the following \Gee sequent, where
$L$ and $E$ are promoted to eigenvariables, and the assumptions of the lemma are
converted to hypotheses.
\begin{smallmultline}
  L, E ;
  (\kctx~L)^@, \kmember{E}{L} ||- {} \\
  \TAB \EX x, H. E = \big(\SALL i, k. \kldb~x~i~(\kdvar~k\bigr)
  \sif \kadd~H~k~i) \rand \kname{x}.
\end{smallmultline}

The \IH cannot be immediately used because the annotations of $\kctx~L$ do not
match.
To make progress, the definition of $\kctx~L$ needs to be \emph{unfolded}.
As explained in \secref{gee}, this amounts to finding all ways of unifying
$\kctx~L$ with the heads of the clauses in the definition of $\kctx$.
The complete set of unifiers is characterized by $L = \knil$ and $L = (\SALL i,
k. \hodb~\kn~i~(\kdvar~k) \sif \kadd~H~k~i) \kcons L'$ for new eigenvariables
$H$ and $L'$ and a nominal constant $\kn$.
In the latter case we also have a new hypothesis, $(\kctx~L')^*$, that comes
from the body of the second clause for $\kctx$.
There are two things to note: first, the $\nabla$ at the head of the second
clause of $\kctx$ is turned into a nominal constant in the proof obligation, and
the second is that the new hypothesis in the second case is annotated with $^*$,
which suits the \IH.

In each case for $L$, the argument proceeds by analyzing the second hypothesis,
$\kmember{E}{L}$.
The case of $L = \knil$ is vacuous, because there is no way to infer
$\kmember{E}{\knil}$, making that hypothesis equivalent to false.
In the case of $L = (\SALL i, k. \hodb~\kn~i~(\kdvar~k) \sif \kadd~H~k~i) \kcons
L'$, we have two possibilities for $\kmember{E}{L}$: either
\[E = (\SALL i, k.
\hodb~\kn~i~(\kdvar~k) \sif \kadd~H~k~i),\]
or $\kmember{E}{L'}$.
The former possibility is exactly the conclusion that we seek, so this branch of
the proof finishes.
The latter possibility lets us apply \IH to the hypotheses $(\kctx~L')^*$ and
$\kmember{E}{L'}$, which also yields the desired conclusion.

The second necessary lemma asserts that there is at most a single clause for
each variable in the dynamic context.
\begin{smallgathertagged}
  \begin{split}
    & \ALL L, x,\ H_1, H_2. \kctx~L \rimp \\
    &\TAB \kmember{(\SALL i, k. \kldb~x~i~(\kdvar~k) \sif \kadd~H_1~k~i  )}{L} \rimp \\
    &\TAB \kmember{(\SALL i, k. \kldb~x~i~(\kdvar~k) \sif \kadd~H_2~k~i)}{L} \rimp \\
    &\TAB H_1 = H_2.
  \end{split}
  \label{lab:hodb-ctx-sync}
\end{smallgathertagged}
Note that from $H_1 = H_2$, we are able to conclude that the two dynamic clauses
relating $x$ to a De Bruijn term must be the same.
Like the previous lemma, it is proved by induction on the hypothesis $\kctx~L$.

Armed with these lemmas, we can then show both directions of determinacy for
\hodb.
In the forward direction the statement is as follows.
\begin{smallmultline}
  \ALL L, M, H, D, E.
  \kctx~L \rimp \\
  \oseq{L |- \hodb~M~H~D} \rimp
  \oseq{L |- \hodb~M~H~E} \rimp D = E.
\end{smallmultline}
We prove this by induction on $\oseq{L |- \hodb~M~H~D}$; this amounts to
assuming the lemma \IH below:
\begin{smallalign}
  \ALL L, M, H, D, E. & \kctx~L \rimp
  \oseq{L |- \hodb~M~H~D}^* \rimp \\
  &\TAB \oseq{L |- \hodb~M~H~E} \rimp D = E \tag{\IH}
\end{smallalign}
and proving the \Gee sequent
\begin{smallalign}
  L, M, H, D, E ; {} &\kctx~L, \oseq{L |- \hodb~M~H~D}^@, \\
  &\TAB \oseq{L |- \hodb~M~H~E} ||- D = E.
\end{smallalign}

Now, $\oseq{L |- \hodb~M~H~D}^@$ is just a notation for the \Gee atom
$\kseq~L~(\hodb~M~H~D)^@$ whose definition is given by the clauses in
\figref{hhw-encoding}.
Unfolding the definition amounts to finding all the clauses in
\figref{hhw-encoding} whose heads match \[\kseq~L~(\hodb~M~H~D).\]
Only the final two clauses of \kseq, corresponding to the rules \rn{dyn} and
\rn{prog} of \HHw, are therefore relevant.

Let us consider backchaining the static clauses first, \ie, the applications of
the \rn{prog} rule.
There are only a finite number of them, so the assumption $\kprog~F$ can be
immediately turned into a branched tree with one case for every static program
clause.
For the first static clause, we are left with a new assumption:
\begin{smallalign}
  \Bigl\{
  L, \Bigl[
  & \SALL M', N', H', D', E'. \hodb~(\kapp~M'~N')~H'~(\kdapp~D'~E') \sif \\
  & \TAB \hodb~M'~H'~D' \sand \hodb~N'~H'~E'
  \Bigr] |- \hodb~M~H~D
  \Bigr\}^*
\end{smallalign}
The annotation $^*$ here was obtained from unfolding the definition of a
$^@$-annotated atom per the technique outlined in \secref{gee}.
Note that this is just a backchaining sequent ($\kbc$) whose definition in
\figref{hhw-encoding} can be unfolded.
Doing this instantiates the $\sfall$ prefix in the bakchaining clause in such a
way that the head $\hodb~(\kapp~M'~N')~H'~(\kdapp~D'~E')$ unifies with the \HHw
formula on the right, $\hodb~M~H~D$; this produces the substitutions $M =
\kapp~M'~N'$, $H = H'$, and $D = \kdapp~D'~E'$ for fresh eigenvariables $M', N',
H', D', E'$.
Moreover, by the second clause for \kbc in \figref{hhw-encoding}, we get
this goal reduction sequent as a fresh hypothesis:
\begin{smallgather}
  \oseq{L |- \hodb~M'~H'~D' \sand \hodb~N'~H'~E'}^*
\end{smallgather}
which is reduced by the first clause for \kseq to:
\begin{smallgather}
  \oseq{L |- \hodb~M'~H'~D'}^* \text{ and } \oseq{L |- \hodb~N'~H'~E'}^*
\end{smallgather}

We can almost apply the induction hypothesis \IH---we know $\kctx~L$ and
$\oseq{L |- \hodb~M'~H'~D'}^*$ already---but we still must find the third
argument.
To get this argument we need to case analyze the other hypothesis, $\oseq{L |-
  \hodb~M~H~E}$, which becomes $\oseq{L |- \hodb~(\kapp~M'~N') ~H' ~E}$ as a
result of the previous unification.
It has no size annotations because the induction was on the first hypothesis.
Nevertheless, we can perform a case analysis of its structure by unfolding its
definition (using the clauses in \figref{hhw-encoding}).
Once again, we have a choice of using a static program clause or a dynamic
clause from $L$.
If we use a static clause, then by a similar argument to the above we will get
the following fresh hypotheses, for new eigenvariables $D''$ and $E''$ such that
$E = \kdapp~D''~E''$:
\begin{smallgather}
  \oseq{L |- \hodb~M'~H' ~D''} \text{ and } \oseq{L |- \hodb~N'~H'~E''}
\end{smallgather}
We can now apply the \IH twice, yielding $D' = D''$ and $E' = E''$, so $D =
\kdapp~D'~E' = \kdapp~D''~E'' = E$.

If, on the other hand, we use a dynamic clause in $L$, then the two fresh
hypotheses we get are:
\begin{smallgather}
  \kmember{F}{L} \text{ and } \oseq{L, \foc{F} |- \hodb~(\kapp~M'~N')~H~E}.
\end{smallgather}
for some new eigenvariable $F$.
This is the first place where the context characterization hypothesis $\kctx~L$
becomes useful.
By Lem.~(\ref{lab:hodb-ctx-inv}) above, we should be able to conclude that $F$
is of the form $\big(\SALL i, k. \kldb~\kn~i~(\kdvar~k) \sif \kadd~\tilde
H~k~i\bigr)$ for some term $\tilde H$ and nominal constant $\kn$.
By looking at the clauses for \kbc in \figref{hhw-encoding}, it is clear
that there is no way to prove the sequent $\oseq{L, \foc{F} |-
  \hodb~(\kapp~M'~N')~H~E}$, because the term $\kn$ will never unify with
$\kapp~M'~N'$.
Hence this hypothesis is vacuous, which closes this branch.
We have now accounted for all the cases of backchaining a static clause for the
inductive assumption $\oseq{L |- \hodb~M~H~D}^@$\kern -3pt.

This leaves only the dynamic clauses in $L$---which are backchained using the
\rn{dyn} rule---which corresponds to the following pair of new hypotheses:
\begin{smallgather}
  \kmember{F}{L}^* \text{ and }
  \oseq{L, \foc{F} |- \hodb~M~H~D}^*
\end{smallgather}
As these hypotheses come from unfolding an inductive assumption, they are
$^*$-annotated.
Once again, we can apply lem.~(\ref{lab:hodb-ctx-inv}) to conclude that
\begin{smallgather}
  F \quad = \quad \big(\SALL i, k. \kldb~\kn~i~(\kdvar~k) \sif \kadd~H~k~i\bigr)
\end{smallgather}
where \kn denotes a nominal constant,
We then continue using the definitional clauses for \kbc to get the fresh
assumption
\begin{smallgather}
  \oseq{L |- \kadd~H~k~i}^*
\end{smallgather}
for new eigenvariables $k$ and $i$, and the equations $M = \kn$ and $D = \kdvar~
k$.
Then, since \kn does not occur in the static clauses $\Th$, the only way to
prove the second hypothesis $\oseq{L |- \hodb~\kn~H~E}$ would be to use a
dynamic clause in $L$.
Once again, by lem~(\ref{lab:hodb-ctx-inv}) and unfolding the definition of \kbc
as above, we see that this clause must have been of the form $\SALL i, k.
\kldb~\kn~i~(\kdvar~k) \sif \kadd~\hat H~k~i$ for some eigenvariable $\hat H$.
We can now use the other lemma~(\ref{lab:hodb-ctx-sync}) to show that $ H = \hat
H$.
Hence, $\oseq{L |- \hodb~\kn~H~E}$ backchains on the same clause in $L$ as
$\oseq{L |- \hodb~\kn~H~D}$, so it must be that $E = \kdvar~k$ as well, \ie, $D
= E$.

\begin{figure}[!tb]
  \bgroup \small
  \underline{\texttt{Theorem}} \texttt{ctx\_inv} :
  $\ALL L, E. \kctx~L \rimp \kmember{E}{L} \rimp {}$ \\
  \TAB $\EX x, H. E = (\SALL i, k. \hodb~x~i~(\kdvar~k) \sif \kadd~H~k~i) \rand
  \kname~x$.
  \\
  \underline{\texttt{Theorem}} \texttt{ctx\_unique} :
  $\ALL L, x, H_1, H_2. \kctx~L \rimp {}$ \\
  \TAB $\kmember{(\SALL i, k. \hodb~x~i~(\kdvar~K) \sif \kadd~H_1~k~i)}{L} \rimp{}$ \\
  \TAB $\kmember{(\SALL i, k. \hodb~x~i~(\kdvar~K) \sif \kadd~H_2~k~i)}{L}
  \rimp H_1 = H_2$.
  \\
  \underline{\texttt{Theorem}} \texttt{add\_det2} :
  $\ALL L, X, Y_1, Y_2, Z. \kctx~L \rimp {}$ \\
  \TAB $\oseq{L |- \kadd~X~Y_1~Z} \rimp \oseq{L |- \kadd~X~Y_2~Z}
  \rimp Y_1 = Y_2$. \\
  \underline{\texttt{Theorem}} \texttt{hodb\_det3} :
  $\ALL L,M,D,E,H. \kctx~L \rimp {}$ \\
  \TAB
  $\oseq{L |- \hodb~M~H~D} \rimp
  \oseq{L |- \hodb~M~H~E} \rimp
  D = E$.
  \\\ttfamily
  induction on 2.\ intros cH dH eH. dcH:case dH. \\
  \TAB \textrm{\itshape \% case of $M = \kapp~M_`~M_2$} \\
  \TAB ecH:case eH.\\
  \TAB[2] apply IH to cH dcH ecH. \\
  \TAB[2] apply IH to cH dcH1 ecH1.\@ search. \\
  \TAB[2] bcH:apply ctx\_inv to cH. case bcH. case ecH. \\
  \TAB \textrm{\itshape \% case of $M = \kabs~M'$} \\
  \TAB ecH:case eH.\\
  \TAB[2] apply IH to \_ dcH ecH. search. \\
  \TAB[2] bcH:apply ctx\_inv to cH ecH1.\ case bcH. case ecH. \\
  \TAB \textrm{\itshape \% backchaining on $L$} \\
  \TAB bcH:apply ctx\_inv to cH dcH1.\ aH:case dcH. \\
  \TAB[2] ecH:case eH. case bcH. case bcH. \\
  \TAB[2] uH:apply ctx\_inv to cH bH1.\ case uH. \\
  \TAB[2] bH:case ecH. apply ctx\_unique to cH dcH1 ecH1. \\
  \TAB[2] apply add\_det2 to cH aH bH search.
  \egroup
  \caption{Abella proof that \hodb is deterministic in its third argument}
  \label{fig:hodb-proof}
\end{figure}

This proof, which has been explained here in some detail, is concisely expressed
using the tactics language of Abella~\cite{abella-hhw} as shown in
\figref{hodb-proof}.
Induction and case analyses are indicated explicitly using the \kwd{induction}
and \kwd{case} tactics, while lemmas are applied using the \kwd{apply} tactic.
The \kseq and \kbc definitions are used implicitly by the \kwd{case} and
\kwd{search} tactics; in particular \kwd{case} handles reasoning on backchaining
sequents.
The tactics language of Abella therefore remains unchanged from earlier versions
that were designed to support only second-order hereditary Harrop formulas.

The \hodb relation is also deterministic in its first argument---\ie, given a De
Bruijn indexed term, there is at most a single \HOAS term it corresponds
to---which is proved in a similar fashion.
Thus, the \hodb relation is manifestly an isomorphism between the two
representations of \lterms.
If we were to specify the translations functionally, then we would not only have
to repeat the clauses for both directions of the translation, but we would also
have to prove separately that they are injective and inverses.
We do not sacrifice any of the executable power of a functional specification:
the program \hodb is directly executable in the language
$\lambda$Prolog~\cite{nadathur88iclp}.

\section{Case Study: Relating Marked Reduction to Lambda Paths}
\label{sec:paths}

\noindent%
The example of the previous section was simple enough that the dynamic context
could always be characterized directly by an inductive definition.
In the general case, we will need to prove properties about a collection of
higher-order relations where each relation has its own separate form of dynamic
context.
We will therefore need to generalize unary definitions such as $\kctx$ of the
previous section to \emph{context relations} of higher arities.
This section contains a case study of such an example, which is independently
novel.

The example is drawn from~\cite[Sec.~7.4.2]{miller12proghol} and involves a
structural characterization of reductions on \lterms.
A \emph{path} through a \lterm is a way to reach any non-binding occurrence of a
variable in the term~\cite[Sec.~4.2]{miller12proghol}.
In \HHw, we can use a basic type $\kp$ for paths with the following
constructors: $\kleft,\kright : \kp \to \kp$ to descend to the function or the
argument sub-trees in an application, and $\kbind : (\kp \to \kp) \to \kp$ to
descend through a $\lambda$-abstraction.
Crucially, $\kbind$ has the same binding structure as the $\lambda$-abstractions
encountered along the path.
The predicate $\kpath : \ktm \to \kpath \to \omic$ asserts that a given \lterm
contains a given path; it is defined by the following three \HHw clauses.
\begin{smallalign}
  &\kpath~(\kapp~M~N)~(\kleft~P) \sif \kpath~M~P. \\
  &\kpath~(\kapp~M~N)~(\kright~P) \sif \kpath~N~P. \\
  &\kpath~(\kabs~M)~(\kbind~P) \sif \\
  &\TAB \SALL x, p. \kpath~(M~x)~(P~p) \sif \kpath~x~p.
\end{smallalign}
As these paths record the specific structure of a \lterm, $\beta$-reduction
changes the paths in the term.
On the other hand, a path through \emph{the result of reducing} $\kapp\app
(\kabs\app (\LAM x. M x))\app N$ would be a path through $M\app x$ with the
additional proviso that any path through $N$ is also a path through $x$.
Paths are a useful tool for structural characterization of terms.
For instance, if two terms have the same paths, then they must be identical;
this corresponds to the following theorem of \Gee:
\begin{smallgathertagged}
  \begin{split}
    & \ALL M, N. \\
    &\TAB (\ALL P. \oseq{{} |- \kpath~M~P} \rimp \oseq{{} |- \kpath~N~P}) \rimp \\
    &\TAB[2] M = N.
  \end{split}
  \label{lab:paths-same-nored}
\end{smallgathertagged}
This theorem is provable in the version of Abella described in~\cite{gacek12jar}
that only has the second-order fragment of \HHw as its specification logic.

Unfortunately, this structural characterization is not preserved by
$\lambda$-conversion.
Suppose we want to compute the paths in a term that result from reducing certain
marked $\beta$-redexes.
Formally, we can add a new constructor for marked redexes, $\kbeta : (\ktm \to
\ktm) \to \ktm \to \ktm$ with the understanding that $\kbeta\app M\app N$
denotes the same \lterm as $\kapp\app (\kabs\app M)\app N$, except that the redex
is marked.
We can then define a relation $\kbred : \ktm \to \ktm \to \omic$ that reduces
all the marked $\beta$-redexes in a term, with the following clauses.
\begin{smallalign}
  &\kbred~(\kapp~M~N)~(\kapp~U~V) \sif \kbred~M~U \sand \kbred~N~V. \\
  &\kbred~(\kabs~M)~(\kabs~U) \sif \\
  &\TAB \SALL x. \kbred~(M~x)~(U~x) \sif \kbred~x~x. \\
  &\kbred~(\kbeta~M~N)~V \sif \\
  &\TAB \SALL x. \kbred~(M~x)~V \sif \SALL u. \kbred~x~u \sif \kbred~N~u.
\end{smallalign}
We also need a static clause for a path in a marked redex.
\begin{smallalign}
  &\kpath~(\kbeta~M~N)~P \sif \\
  &\TAB \SALL x. \kpath~(M~x)~P \sif \SALL q. \kpath~x~q \sif \kpath~N~q.
\end{smallalign}
Since different terms can have the same paths as long as they reduce to the same
term, the theorem (\ref{lab:paths-same-nored}) will need to be updated to
account for reduction.
That is, if two terms have the same paths, then they are \emph{joinable} by
\kbred:
\begin{smallgathertagged}
  \begin{split}
    &\ALL M, N, U, V. \\
    &\TAB (\ALL P. \oseq{{} |- \kpath~M~P} \rimp \oseq{{} |- \kpath~N~P}) \rimp \\
    &\TAB[2]  \oseq{{} |- \kbred~M~U} \rimp \oseq{{} |- \kbred~N~V} \rimp U = V.
  \end{split}
  \label{lab:paths-joinable}
\end{smallgathertagged}

How would one prove (\ref{lab:paths-joinable})?
Note that there are two different higher-order predicates: proofs of $\kbred\app
M\app U$ will add dynamic clauses involving $\kbred$, while proofs of
$\kpath\app M\app P$ will add dynamic clauses involving $\kpath$.
We would like to prove that $\kbred$ preserves $\kpath$, so the statement of the
theorem would have to account for proofs of both kinds, and hence for both kinds
of dynamic clauses.
The general technique in \Gee for such situations is to \emph{relate} the two
kinds of dynamic contexts for the two different relations.
The following definition of $\kctxx : \olist \to \olist \to \prop$ achieves
this.
\begin{smallalign}
  & \kctxx~\knil~\knil \rdef \rtrue \\
  & \Bigl(
  \NABx x, p. \kctxx~(\kbred~x~x \kcons K)~(\kpath~x~p \kcons L)
  \Bigr) \rdef \kctxx~K~L ; \\
  & \Bigl(\NABx x.
     \kctxx~((\SALL u. \kbred~N~u \simp \kbred~x~u) \kcons K) \\
     & \TAB[2] \TAB ((\SALL p. \kpath~N~p \simp \kpath~x~p) \kcons L)
  \Bigr) \rdef \kctxx~K~L.
\end{smallalign}

It is important to note that the $\kctxx$ predicate not only says how two such
contexts are related, but also contains a specification of the contexts themselves.
A hypothesis $\kctxx~K~L$ where $L$, say, is not used elsewhere in the theorem
is equivalent to assuming just that $K$ is a dynamic context for $\kbred$.
As before, the $\rnbl$-bound variables at the head guarantee that every such
variable has a unique dynamic clause in both contexts, which we can establish
separately using a lemma.

The proof of (\ref{lab:paths-joinable}) now proceeds as follows: first we note
that if $\kbred~M~N$, then a path in $M$ must also be in $N$ and \emph{vice
  versa}.
Then, we separately show that if $\kbred~M~N$, then it must be that $N$ is free
of any subterms involving $\kbeta$.
Finally, we prove the lemma that if two $\kbeta$-free terms have the same paths,
then they must be identical, which is essentially the same theorem as
(\ref{lab:paths-same-nored}).

Let us consider the first of these lemmas: that $\kbred$ preserves $\kpath$.
In the \Gee encoding of \HHw, the statement of the theorem is:
\begin{smallgathertagged}
  \begin{split}
    & \ALL K, L, M, U, P. \\
    &\TAB \kctxx~K~L \rimp \oseq{K |- \kbred~M~U} \rimp \\
    &\TAB[2] \oseq{L |- \kpath~M~P} \rimp \oseq{L |- \kpath~U~P}.
  \end{split}
  \label{lab:paths-preserve-l2r}
\end{smallgathertagged}
This theorem is proved by induction on $\oseq{K |- \kbred~M~U}$.
Just as in the inductive proofs in \secref{hodb}, there will be some cases
for backchaining static program clauses and some for dynamic clauses.
The static cases are fairly straightforward, so we concentrate below on the
dynamic cases.

Per the definition in \figref{hhw-encoding}, backchaining a dynamic clause
for $\oseq{K |- \kbred~M~U}$ produces the new hypotheses:
\begin{smallgather}
  (\kmember{E}{K})^* \quad \text{and}\quad \oseq{K, \foc{E} |- \kbred~M~U}^*
\end{smallgather}
for some eigenvariable $E$.
From $\kctxx~K~L$ and $\kmember{E}{K}$, it must follow that:
\begin{smallalign}
  & (\EX X. (E = \kbred~X~X) \rand \kname X) \\
  &\TAB \ror (\EX N, X. E = (\SALL u. \kbred~X~u \sif \kbred~N~u) \rand \kname X)
\end{smallalign}
which is itself proven (as a lemma) by induction on the hypothesis $\kctxx~K~L$.
We therefore need to consider only these two cases for the dynamic clause $E$.

The first case where $E = \kbred~X~X$ is easy to prove.
For the second case, we are left with the following problem: although we
can characterize the cases for $E$, this is not enough to reason about $\kpath$
because $E$ is a dynamic clause for $\kbred$.
This is where we use the fact that $\kctxx$ is a relation to prove the following
lemma.
\begin{smallgathertagged}
  \begin{split}
    & \ALL K, L, N. \NABx n.\ \kctxx~(K~n)~(L~n) \rimp \\
    &\TAB \kmember{(\SALL u. \kbred~n~u \sif \kbred~N~u)}{(K~n)} \rimp \\
    &\TAB[2]
    \kmember{(\SALL q. \kpath~n~q \sif \kpath~N~p)}{(L~n)}
  \end{split}
  \label{lab:paths-ctx2-sync}
\end{smallgathertagged}
Its proof is by induction on the hypothesis $\kctxx~K~L$.
It can be seen as a kind of translation between the formal relation, given as an
inductive definition, to a way of reasoning about the elements of the related
contexts.
The lemma (\ref{lab:paths-ctx2-sync}) states, in particular, that a dynamic
clause about reduction of marked redexes in the dynamic contexts for $\kbred$
must have a corresponding dynamic clause for paths through a marked redex in the
dynamic contexts for $\kpath$.

We now have nearly everything to finish the proof of
(\ref{lab:paths-preserve-l2r}).
The only remaining wrinkle is that in the case where the term $M$ is a variable
that unifies with a nominal constant $n$, we will need to look up its dynamic
clause in a suitable dynamic context and continue by back\-chaining it.
This amounts to the following \emph{inversion lemma}:
\begin{smallgathertagged}
  \begin{split}
    & \ALL K, L, N, P. \NABx n.\ \kctxx~(K~n)~(L~n) \rimp \\
    &\TAB \kmember{(\SALL q. \kpath~n~q \sif \kpath~N~q)}{(L~n)} \rimp \\
    &\TAB[2] \oseq{(L~n) |- \kpath~n~P} \rimp \oseq{(L~n) |- \kpath~N~P}.
  \end{split}
  \label{lab:paths-inversion}
\end{smallgathertagged}
Effectively, this lemma says that the only way that $\oseq{(L~n) |- \kpath~n~P}$
could have been proved is by backchaining on the given clause, which has the
premise $\oseq{(L~n) |- \kpath~N~P}$.
We can show this lemma because we have completely characterized the dynamic
context $(L~n)$, and the static program has no clauses with nominal constants.
Note that the nesting order of $\forall$ and $\nabla$ is crucial here: the
nominal constant $n$ must not be allowed to occur in $N$.
However, it is obviously allowed to occur in the dynamic context, so we indicate
this by means of an explicit dependency, indicated here using the application
$(L~n)$.
This punning between the two levels is possibly because \HHw and \Gee are both
based on a common \lcalculus.
The proof of (\ref{lab:paths-preserve-l2r}) can now be completed by using
(\ref{lab:paths-inversion}) for the variable case.

The full development of this example in Abella, including the formal proofs, can
be found in \url{examples/hhw/breduce.thm} in the Abella
distribution~\cite{abella-hhw}.
%

\section{Related Work}
\label{sec:related}

\noindent%
The \HHw proof system presented in \secref{hhw} is largely similar to the
focused sequent calculus \LJF~\cite{liang09tcs} for the fragment of
intuitionistic logic containing implication, universal quantification,
negatively polarized atoms, and the negatively polarized variant of conjunction.
It is also straightforwardly a version of the calculus formalizing \emph{uniform provability}~\cite{miller91apal}.
The term ``logic of hereditary Harrop formulas'' is often used to indicate an
extended logic where disjunction and existential quantification are also allowed
in a limited form~\cite[chap.~3]{miller12proghol}.
Specifications in the full language with these connectives can be compiled into
our \HHw language, possibly with an increase in the number of static clauses in
the specifications.

Representational techniques for data with binding can be broadly classified into
two styles: first-order and higher-order.
Regardless of style, a primary requirement of the representation is that it not
distinguish between terms that are $\alpha$-equivalent.
The traditional first-order approach to realizing this requirement is to
represent bound variables by De Bruijn indexes, which yields canonical
representatives of $\alpha$-equivalence classes of \lterms.
A very different first-order alternative to De Bruijn indexes is the approach of
nominal logic that forgoes canonical representatives of the $\alpha$-equivalence
classes; instead, two terms are considered identical if they are
\emph{equivariant}, meaning that the names used in one term can be permuted to
the names in the other.
This approach is the basis of Nominal Isabelle~\cite{cheney08toplas}, and there
are also a number of libraries for programming with nominal data, such as Fresh
OCaml~\cite{gabbay03icfp} and Alpha Prolog~\cite{cheney04iclp,urban05tlca}.

A drawback with first-order representations, whether of the De Bruijn or the
nominal logic kind, is that they typically do not offer support for binding
related notions beyond $\alpha$-equivalence.
Typical reasoning applications require a realization of operations such as
substitution and analysis of syntactic structure that respects binding.
With first-order approaches, these have to be implemented explicitly and the
reasoning process must also show their correctness.
In particular, the operation of substitution of a term for a free variable,
which is at the heart of much of the meta-theory of deductive systems, requires
careful book-keeping and fairly detailed correctness arguments (see
\eg~\cite{polonowski13itp} for a recent example done in Coq).
In contrast, higher-order representations reflect binding constructs into the
meta-level abstraction operation and thereby absorb arguments about the
correctness of binding related operations into a one-time argument, external to
the object-level reasoning task, about the correctness of the the meta-language
implementation.

Besides Abella, there are three other systems designed to reason about
specifications in \HOAS: \Hybrid~\cite{felty12jar},
\Beluga~\cite{pientka10ijcar}, and \Twelf~\cite{pfenning99cade}.
All of these systems are broadly two-level or nested systems, but they make
different choices for the specification and reasoning formalisms.
Of these, only \Hybrid is integrated with popular existing formal reasoning
systems (Coq and Isabelle), which allows it to leverage the trusted kernels of
the existing systems instead of implementing new trusted components.
On the other hand, \Hybrid is limited to the second-order hereditary Harrop
fragment for the specification level (which makes it similar in this respect to
the earlier version of Abella described in~\cite{gacek12jar}) and does not have
support for generic reasoning.
The second-order restriction is significant when reasoning about higher-order
deductive systems: the dynamic clauses of higher-order specifications must be
named and transferred to the static program beforehand.
For example, the \kpath predicate in the second-order fragment requires an
auxiliary predicate \kjump and the following clauses for marked redexes.
\begin{smallalign}
  &\kpath~(\kbeta~M~N)~P \sif
  \SALL x. \kpath~(M~x)~P \sif \kjump~x~N. \\[1ex]
  &\kpath~X~P \sif \kjump~X~N \sand \kpath~N~P.
\end{smallalign}
Writing such auxiliary predicates is not only error-prone and anti-modular, but
they also complicate reasoning about the relations.
For instance, in \HHw it is a direct consequence of cut that $\oseq{L,
  \kpath~N~q |- \kpath~x~q}$ and $\oseq{L, (\SALL q. \kpath~x~q \sif \kpath~N~q)
  |- G}$ imply $\oseq{L |- G}$.
However, for the second-order encoding above, the fact that $\oseq{L, \kpath~N~q
  |- \kpath~x~q}$ and $\oseq{L, \kjump~x~N |- G}$ imply $\oseq{L |- G}$ would
need a separate inductive proof.

\Beluga and \Twelf both use the \LF dependent type theory for their
specification languages.
It is known that \LF specifications can be systematically and faithfully
translated into \HHw~\cite{felty90cade,snow10ppdp}.
The encoding of an \LF signature in \HHw uses higher-order features
pervasively, and, indeed, was an early motivation for the present work of
supporting reasoning over higher-order specifications in Abella.
The main difference between \LF and \HHw is their type systems, which directly
affects their reasoning principles.
Briefly, \LF encourages a ``combined contexts'' reasoning approach, while \HHw
encourages a ``context relations'' approach.
Because \LF is dependently typed, the dynamic signature extensions for
universally quantified goals cannot be separated from other assumptions; in
fact, contexts in \LF are interpreted as ordered.
It is difficult to place the same \LF term in two different
contexts.

In both \Beluga and \Twelf, therefore, the most direct way to reason about
different higher-order relations is to use a common dynamic context for the
relations.
This is achieved formally by specifying contexts \emph{schematically} by means
of regular grammars, and using \emph{subordination} analysis on the signature to
determine when one regular context may be \emph{subsumed} by another.
For example, since there is no way to embed a \lterm value inside a \knat value
using the provided constructors, and the clauses for \kadd do not mention
\lterms, it must stand to reason that properties of \kadd must hold even in a
context of assumptions about \lterms.
Such subsumption properties are often useful; for examples, in the example of
\secref{hodb}, if the required properties of \kadd are used in a non-empty
dynamic context, we must separately prove that earlier theorems still hold, such
as the theorem \texttt{add\_det2} in \figref{hodb-proof}.
In Abella, context definitions are no different from any other inductive
definition; there is no automatic subsumption of context relations and such
lemmas must be proven manually.
On the other hand, reasoning about contexts is not part of the trusted base of
Abella, and many properties about arbitrary context relations can be separately
proved and used in a modular fashion, as we have done in the examples in
\secref{hodb} and \ref{sec:paths}.

The differences between the Abella approach and that of \Twelf and \Beluga taken
together can be summarized by the following observations.
Firstly, \Twelf and \Beluga make many kinds of reasoning about context
membership, such as (\ref{lab:hodb-ctx-inv}), automatic and available to the user
for free.
Explicit reasoning about context members in Abella can be tedious, so it is
conceivable that some aspects of the context reasoning of \Twelf and \Beluga can
be imported into Abella in the future.
In particular, theorems such as (\ref{lab:hodb-ctx-inv}),
(\ref{lab:hodb-ctx-sync}), and (\ref{lab:paths-ctx2-sync}) have entirely
predictable proofs that should be easy to automate.

Secondly, the reasoning logic \Gee has a well-developed proof-theory that
includes a sequent calculus with a cut-admis\-sibility result~\cite{gacek11ic}.
This logic has a number of features: an equality predicate at all types, generic
reasoning, and both inductive and co-inductive fixed-point definitions.
\Twelf's \MTwoPlus meta-logic also has a sequent calculus with proof-terms, and
the consistency of this logic is proved by giving the proof terms an operational
semantics and verifying that they represent total functions under this
interpretation.
\Beluga (as of version 0.5) supports only inductive reasoning in terms of
recursive fixed-points, and does not support co-induction.
\Twelf supports inductive reasoning for $\Pi_0^1$ theorems, but also has no
support for co-induction.
Neither \Twelf nor \Beluga has a built-in equality predicate.
For generic reasoning, \Beluga's contextual modal types can achieve many of the
same goals as the $\rnbl$ quantifier of $\Gee$, but the global nature of nominal
constants and equivariant unification makes it possible to reason about open
terms with free variables, unaccompanied by any contexts~\cite{accattoli12cpp}.
Much of the informal meta-theory of the \lcalculus uses open terms in this
style, but a first order representation of variables requires an explicit
treatment of $\alpha$-equivalence and substitution.
The $\rnbl$ quantifier lets us combine the benefits of \HOAS and reasoning on
open terms.

\def\kistm{\kwd{is\_tm}}

Finally, the type systems of \Twelf and \Beluga are endowed with an associated
natural induction principle that allows reasoning by induction on the structure
of well-typed terms.
In Abella, typing is not treated as a definition, so if one wants to induct on
the structure of \lterms, for example, one would have to use a well-formedness
predicate $\kistm : \ktm \to \omic$ with the following clauses:
\begin{smallalign}
  &\kistm~(\kapp~M~N) \sif \kistm~M \sand \kistm~N. \\
  &\kistm~(\kabs~M) \sif \SALL x. \kistm~(M~x) \sif \kistm~x.
\end{smallalign}
Then, whenever one needs to reason by induction on the structure of a term $M$,
one reasons instead on $\oseq{{} |- \kistm~M}$.
Because such predicates essentially reify the well-typedness relation, they will
generally need higher-order clauses if the types of the constructors are
higher-order.
For instance, the \kabs constructor has a second-order type and requires a
second-order clause for \kistm.
Note that such definitions cannot be made in the reasoning logic \Gee because
they are not stratified, \ie, to prove $\kistm~M$, one needs to make assumptions
of the form $\kistm~x$.
It is of course possible to automatically generate \HHw predicates like \kistm
for a given \HHw signature, but in any theorem that involves inductive reasoning
on the structure of terms one would still need to make hypotheses such as
$\oseq{{} |- \kistm~M}$ explicit.
Note that $\ALL M{:}\ktm. \oseq{{} |- \kistm~M}$ is not a theorem of \Gee.

\section{Conclusion}
\label{sec:conclusion}

\noindent%
We have presented an extension to the two-level logic approach that lets one use
the full richness of \HHw to specify and formally reason about higher-order
deductive formalisms.
The essence of our method is characterizing the contexts of these higher-order
formalisms as inductive relations, and a variant of the backchaining procedure
that allows us to use properties of these inductive characterizations in a
modular way.
We have validated our design and methodology by implementing an extended Abella
system and by using it to develop a number of non-trivial examples of reasoning
over higher-order specifications.

\medskip

\noindent%
\emph{Acknowledgments}:
We thank Dale Miller, Olivier Savary-B{\'e}langer and the anonymous reviewers
for helpful discussions and comments on earlier drafts.
This work has been partially supported by the NSF Grants OISE-1045885 (REUSSI-2)
and CCF-0917140 and by the INRIA Associated Team RAPT.
Opinions, findings, and conclusions or recommendations expressed in this paper
are those of the authors and do not necessarily reflect the views of the
National Science Foundation.

\end{document}